\documentclass[12pt,letter,doublespace]{article}
\usepackage{amssymb}
\usepackage{amsmath} 
\usepackage{graphics}
\usepackage{epsfig}
\usepackage{graphicx}
\usepackage{cite}
\usepackage{amsthm}
\usepackage{color}
\usepackage{booktabs}

\newtheorem{thm}{Theorem}[section]
\newtheorem{lemma}{Lemma}[section]

\newcommand{\vx}{{\bf x}}
\newcommand{\vX}{{\bf X}}
\newcommand{\vDelta}{\mbox{\boldmath $\Delta$}}

\newcommand{\vbeta}{\mbox{\boldmath $\beta$}}

\newcommand{\vxi}{\mbox{\boldmath $\xi$}}
\newcommand{\vB}{{\bf B}}
\newcommand{\vm}{{\bf m}}
\newcommand{\Var}{\mbox{Var}}

\newcommand{\pkonv}{\stackrel{p}{\rightarrow}}
\newcommand{\dkonv}{\stackrel{d}{\rightarrow}}
\newcommand{\askonv}{\stackrel{a.s.}{\longrightarrow}}

\newcommand{\bqa}{\begin{eqnarray*}}
\newcommand{\eqa}{\end{eqnarray*}}
\newcommand{\bqan}{\begin{eqnarray}}
\newcommand{\eqan}{\end{eqnarray}}
\newcommand{\bit}{\begin{itemize}}
\newcommand{\eit}{\end{itemize}}
\newcommand{\ben}{\begin{enumerate}}
\newcommand{\een}{\end{enumerate}}
\newcommand{\beq}{\begin{equation}}
\newcommand{\eeq}{\end{equation}}
\newcommand{\bdes}{\begin{description}}
\newcommand{\edes}{\end{description}}

\renewcommand{\baselinestretch}{1.7}

\begin{document}

\title{Nonparametric Model Checking and Variable Selection}
\author{\sc Adriano Zambom and Michael Akritas\\
The Pennsylvania State University}
\date{December 21, 2011}
\maketitle{}
    \vspace{-2cm}
\section*{Abstract}
\indent \indent Let $\vX$ be a $d$ dimensional vector of covariates and $Y$ be the response variable. Under the nonparametric model $Y=m(\vX)+\sigma(\vX)\epsilon$ we develop an ANOVA-type test for the null hypothesis that a particular coordinate of $\vX$ has no influence on the regression function. The asymptotic distribution of the test statistic, using residuals based on Nadaraya-Watson type kernel estimator and $d\le 4$,  is established under the null hypothesis and local alternatives. Simulations suggest that under a sparse model, the applicability of the test extends to arbitrary $d$ through sufficient dimension reduction. Using p-values from this test, a variable selection method based on multiple testing ideas is proposed.  The proposed test outperforms existing procedures, while  additional simulations reveal that the proposed variable selection method performs competitively against well established procedures. A real data set is analyzed.

\vspace{.5cm}


\textbf{Keywords:} Nonparametric regression; kernel regression; Lack-of-fit tests;
Dimension reduction; Backward elimination.

\vspace{2mm}
 
\noindent{\bf Acknowledgments}: This research was partially supported by CAPES/Fulbright grant 15087657 and NSF grant DMS-0805598.


\pagestyle{plain}
\setcounter{page}{1}
\setlength{\textheight}{9.0in}
\setlength{\topmargin}{-0.5in}

\section{Introduction}  \label{sec:intro}
\indent \indent For a response variable $Y$ and a $d$ dimensional vector of the available covariates $\vX$ set $m(\vX)=E(Y|\vX)$. 
The dual problems of testing for the significance of a particular covariate, and identification of the set of relevant covariates are very common both in applied research and in methodological investigations. Due to readily available software, these tasks are often performed under the assumption of a  linear model, $m(\vX)=\vX\vbeta$. Model checking fits naturally in the methodological context of hypothesis testing, while variable selection is typically addressed through minimization of a constrained or penalized objective function, such as Tibshirani's (1996) LASSO, Fan and Li's (2001) SCAD, Efron, Hastie, Johnstone and Tibshirani's (2004) least angle regression, Zou's (2006) adaptive LASSO, and 
Candes and Tao's (2007) Dantzig selector.

At a conceptual level, however, the two problems are intimately connected: dropping variable $j$ from the model is equivalent to not rejecting the null hypothesis $H_0^j: \beta_j=0$. Abramovich, Benjamini, Donoho and Johnstone (2006) bridged the methodological divide by showing that application of the false discovery rate (FDR) controlling procedure of Benjamini and Hochberg (1995) on $p$ values resulting from testing each $H_0^j$ can be translated into minimizing a model selection criterion of the form 
 \bqan\label{Objective_fun.Penal_LS}
 \sum_{i=1}^n\left(Y_i - \sum_{j\in S} \widehat\beta_j^Sx_{ij}\right)^2 + \sigma^2 |S|\lambda,
 \eqan
where $S$ is a subset of $\{1,2,\ldots,d\}$ specifying the model,  
$\widehat\beta_i^S$ denotes the least squares estimator from fitting model $S$, $|S|$ is the cardinality of the subset $S$, and the penalty 
parameter $\lambda$ depends both on $d$ and $|S|$. This is similar to penalty parameters 
used in Tibshirani and Knight (1999), Birge and Massart (2001) and Foster and Stine (2004), which also depend on both 
$d$ and $|S|$, and more flexible than the proposal in Donoho and Johnstone (1994) which uses $\lambda$ depending only on $d$, as well as AIC and Mallow's 
C$_p$ which use constant $\lambda$. 

Working with orthogonal designs, Abramovich et al. (2006) showed that the global minimum of the penalized least squares 
(\ref{Objective_fun.Penal_LS}) with the FDR penalty parameter
is asymptotically minimax for $\ell^r$ loss, $0<r\le 2$, simultaneously throughout a range of sparsity classes, provided the level $q$ for the FDR is set to $q<0.5$. Generalizations of this methodology to non-orthogonal designs differ mainly in the generation of the $p$ values for testing $H_0^j: \beta_j=0$, and the FDR method employed. Bunea, Wegkamp and Auguste  (2006) use $p$ values generated from the standardized regression coefficients resulting from fitting the full model and employ Benjamini and Yekuteli's (2001) method for controlling FDR under dependency, while  Benjamini and Gavrilov (2009) use $p$ values from a forward selection procedure where the $i$th stage $p$-to-enter is the $i$th stage constant in the multiple-stage FDR procedure in Benjamini, Krieger and Yekutieli (2006).

Model checking and variable selection procedures based on the assumption that the regression function is linear may fail to discern the relevance of covariates whose effect on $m(\vx)$ is nonlinear. See Tables \ref{table:stat1_LV_FL} and \ref{table:LASSO2_sin}.  Because of this, procedures for both model checking and variable selection have been developed under more general/flexible models. See, for example Li and Liang (2008), Wang and Xia (2008), Huang, Horowitz and Wei (2010), 
Storlie, Bondell, Reich and Zhang (2011), and references therein.
 However, the methodological approaches for variable selection under these more flexible models have been distinct from those of model checking. 
 
 This paper aims at showing that a suitably flexible and powerful 
nonparametric model checking procedure can be used to construct a competitive
nonparametric variable selection procedure by exploiting the aforementioned conceptual connection between model checking and variable selection. Thus, this paper has two objectives. The first is to develop a procedure for testing whether a
particular covariate contributes to the regression function in the context of 
a heteroscedastic nonparametric regression model. The second objective is to propose a variable selection procedure based on backward elimination using the 
Benjamini and Yekuteli (2001) method 
applied on the $d$ p-values resulting from testing for the significance of each covariate.

\indent In Section \ref{sec:main}, we formally describe the model and introduce the hypothesis, the test statistic and its asymptotic distribution under the null hypothesis and local alternatives. Section \ref{sec:simulations:model_checking} presents the results of a simulation study where the performance of the proposed test statistic is compared to those of existing tests. In Section \ref{section:Variable_Selection} the proposed variable selection procedure is described and compared, in simulation studies and a real data set, to well established variable selection methods. 

\vspace{-3mm}

\section{Nonparametric Model Checking}  \label{sec:main}
\vspace{-2mm}
\subsection{The Hypothesis and the Test Statistic}\label{test.stat}

\indent \indent Let $Y$ be the response variable and $\vX=(X_1,\ldots,X_d)$ the vector of  available covariates. Set $m(\vX)=E(Y|\vX)$ for the regression function and define 
\bqan\label{rel.true.res}
\zeta=Y-m(\vX). 
\eqan
From its definition if follows that $E(\zeta|\vX)=E(\zeta)=E(\zeta|X_j)=0$, for all $j=1,\ldots,d$. Setting $\sigma^2(\vX)=\Var(\zeta|\vX)$, we have the model
 \bqan\label{rel.np.model}
 Y=m(\vX)+\sigma(\vX)\epsilon,
 \eqan
where $\epsilon$ is the standardized error $\zeta$.
Based on a sample $(Y_i,\vX_i), i=1,\ldots,n$,
of iid observations from model (\ref{rel.np.model}), we will consider testing the hypothesis that the regression function does not depend on the $j$th covariate. For simplicity in notation we set $\vX=(\vX_1,X_2)$, where $\vX_1$ is of dimension $(d-1)$ and $X_2$ is univariate. Setting $E(Y|\vX_1)=m_1(\vX_1)$ the hypothesis we will consider can be written as
\bqan \label{rel.H0} 
H_0: m(\vx_1,x_2) = m_1(\vx_1).
\eqan
\indent To fully appreciate the nature of this hypothesis, let $F_{\vX_1}, F_{X_2}$ denote the marginal distribution functions of $\vX_1, X_2$, respectively, and consider the ANOVA-type decomposition 
\bqan \label{m.decomposition}
m(\vX_1,X_2) = \mu + \tilde{m}_1(\vX_1) + \tilde{m}_2(X_2) + \tilde{m}_{12}(\vX_{1},X_{2}),
\eqan
where
$
  \mu = \int\int m(\vx_{1},x_{2})dF_{\vX_1}(\vx_1)dF_{X_2}(x_{2})$,
 $ \tilde{m}_1(\vx_{1}) = \int m(\vx_{1},x_{2}) dF_{X_2}(x_2) - \mu,$
$  \tilde{m}_2(x_{2}) = \int m(\vx_{1},x_{2}) dF_{\vX_1}(\vx_1) - \mu$,  
$  \tilde{m}_{12}(\vx_{1},x_{2}) = m(\vx_{1},x_{2}) - \mu - \tilde{m}_1(\vx_{1}) - \tilde{m}_2(x_{2})$. Note that their definition implies
$
  \int \tilde{m}_1(\vx_{1})dF_{\vX_1}(\vx_1)=  \int \tilde{m}_2(x_{2})dF_{X_2}(x_2) =  \int \tilde{m}_{12}(\vx_{1},x_{2})dF_{\vX_1}(\vx_1) =
  \int \tilde{m}_{12}(\vx_{1},x_{2})dF_{X_2}(x_2) = 0$.  

\indent Under the null hypothesis (\ref{rel.H0}) it further follows that  
 \bqa
m_1(\vx_1)=\mu+\tilde m_1(\vx_1), \ \ \tilde m_2(X_2)=\tilde m_{12}(\vX_1,X_2)=0.
 \eqa
In the case that $\vX_1,X_2$ are independent, we also have 
$E(Y|X_2)=\mu$ under the null.


\indent  Let now $m_1(\vX_{1i})=E(Y|\vX_{1i})$, as before, and define the null hypothesis residuals as
\bqan\label{rel.H0.res}
\xi_i=Y_i-m_1(\vX_{1i}).
\eqan
Since under the null hypothesis
(\ref{rel.H0}) $m_1(\vX_{1i})=m(\vX_i)$, it follows that the null hypothesis residuals in (\ref{rel.H0.res}) equal the residuals defined in (\ref{rel.true.res}) and thus
\bqan\label{rel.H0.on.res}
E\left(\xi_i|X_{2i}\right)=0.
\eqan
The idea for constructing the test statistic is to think of the $\xi_i$ as data from a high-dimensional one-way ANOVA design with levels $x_{2i}, i=1,\ldots,n$. Because of (\ref{rel.H0.on.res}), it follows that under the null hypothesis (\ref{rel.H0}) there are no factor effects, and we can use the high-dimensional one-way ANOVA statistic of Akritas and Papadatos (2004) for testing (\ref{rel.H0}), after dealing with two important details. First, $m_1$ is not known and needs to be estimated. Second, the statistic of Akritas and Papadatos (2004) requires two or more observations per cell, but in regression designs we typically have only one response per covariate value. 
     
 To deal with the unknow $m_1$ will use the Nadaraya-Watson kernel estimator, 
 \bqan\label{rel.NW}
 \hat{m}_1(\vX_{1i}) = \sum_{j=1}^{n}\left( \frac{K_{H_n}\left(\textbf{X}_{1i}-\textbf{X}_{1j}\right)}{\sum_{l=1}^{n}K_{H_n}\left(\textbf{X}_{1i}-\textbf{X}_{1l}\right)}\right)Y_j,\ i = 1,\ldots,n,
  \eqan
with $K_{H_n}(\vx) = |H_n|^{-1}K(H_n^{-1}\textbf{x})$, where $K(\cdot)$ is a bounded $(d-1)$-variate kernel function of bounded variation and with bounded 
support, and $H_n$ is a symmetric positive definite $(d-1)\times (d-1)$ matrix called the bandwidth matrix. Set
\bqa
\hat\xi_i = Y_i-\hat m_1(\vX_{1i})
\eqa
for the estimated null hypothesis residuals.

To deal with the requirement of more than one observation per cell we make use of smoothness conditions and augment each cell  by including additional $p-1$ $\hat\xi_\ell$'s which correspond to the $(p-1)/2$ $X_{2\ell}$ values that are nearest to $X_{2i}$ on either side. To be specific, we consider the $(\hat\xi_i,X_{2i})$, $i=1,\ldots,n$, arranged so that $X_{2i_1}<X_{2i_2}$ whenever 
$i_1<i_2$, and for each $X_{2i}$, $(p-1)/2<i\le n-(p-1)/2$, define the nearest neighbor window $W_i$ as  
 \bqan\label{def.Wi}
W_i = \left\{j:|\hat{F}_{X_2}(X_{2j}) - \hat{F}_{X_2}(X_{2i})| \leq 
\frac{p - 1}{2n}\right\},
 \eqan
where $\hat{F}_{X_2}$ is the empirical distribution function of $X_2$. $W_i$ defines the augmented cell corresponding to $X_{2i}$. Note that the augmented cells are defined as sets of indices rather than as sets of $\hat \xi_i$ values. 
The vector of $(n-p+1)p$ constructed "observations" in the augmented one-way ANOVA design is  
 \bqan\label{hat.vxi}
 \hat \vxi_V = (\hat{\xi}_j, j \in W_{(p-1)/2+1}, \ldots , \hat{\xi}_j, j \in W_{n-(p-1)/2})' .
 \eqan
\indent Let $\mbox{MST}=\mbox{MST}(\hat\vxi_V)$, $\mbox{MSE}=\mbox{MSE}(\hat\vxi_V)$
denote the balanced one-way ANOVA mean squares due to treatment and error, respectively, computed on the data $\hat\vxi_V$. The proposed test statistic is based on
 \bqan\label{rel.proposedTS}
MST-MSE.
 \eqan
 
\subsection{Asymptotic results}
 
 \subsubsection{Asymptotic null distribution}
 
 \begin{thm} \label{thm:main_X1} 
Assume that the marginal densities $f_{\vX_1}$, 
$f_{X_2}$ of $\vX_1$, $X_2$, respectively, are bounded away from zero, the second derivatives of $f_{\vX_1}$ and $m_1(\vx)$ 
are uniformly continuous and bounded, that 
$\sigma^2(.,x_{2}) := E(\xi^2|X_2=x_2)$ is Lipschitz continuous,  $\sup_{\vx}\sigma^2(\vx)<\infty$, and
$E(\epsilon^4_i)<\infty$. 
Assume that the eigenvalues, $\lambda_i,\  i=1,\ldots,d-1$, of the bandwidth matrix $H_n$
defined in (\ref{rel.NW}), converge to zero at the same rate and satisfy
\bqan\label{cond.d_le_4}
n\lambda_i^8\to 0\ \mbox{ and }\ \frac{n\lambda_i^{2(d-1)}}{(\log n)^2}\to \infty,\  i=1,\ldots,d-1. 
\eqan
Then, under $H_0$ in (\ref{rel.H0}), the asymptotic distribution of the test statistic
in (\ref{rel.proposedTS}) is given by
$$n^{1/2} (MST-MSE) \dkonv N(0,\frac{2p(2p-1)}{3(p-1)}\tau^2),$$
where $\tau = \int\left[\int\sigma^2(\textbf{x}_{1},x_2)f_{\textbf{X}_{1}|X_2=x_2}(\textbf{x}_{1})d\textbf{x}_{1}\right]^2f_{X_2}(x_2)dx_2$.
\end{thm}

\indent An estimate of $\tau^2$ can be obtained by modifying Rice's (1984) estimator as follows
\begin{equation}
\hat{\tau}^2 = \frac{1}{4(n-3)}\sum_{j=2}^{n-2}(\hat{\xi}_j-\hat{\xi}_{j-1})^2(\hat{\xi}_{j+2}-\hat{\xi}_{j+1})^2.
\end{equation}

\indent The next subsection gives the asymptotic theory under local additive and under general local alternatives. As these limiting results show, the asymptotic mean 
of the test statistic $MST - MSE$ is positive under alternatives. Thus, the test procedure rejects the null hypothesis for "large" values of the test statistic. 

\vspace{-2mm}
 
 \subsubsection{Asymptotics under local 
 alternatives}
 
 \vspace{-2mm}
 
\indent \indent The local additive alternatives and the general local alternatives are of the form
\begin{eqnarray} \label{additive_alternative}
H_1^A: m(\vx_1,x_2) &=& m_1(\vx_1) + \rho_n  \tilde m_2(x_2), \\
H^G_1: m(\vx_1,x_2) &=& m_1(\vx_1) + \rho_{1n}\tilde m_2(x_2) + \rho_{2n} \tilde m_{12}(\vx_1,x_2), \label{general_alternative}
\end{eqnarray}
where the functions $\tilde m_2,\ \tilde m_{12}$ satisfy $E\left(\tilde m_2(X_2)\right)=0=E\left(\tilde m_{12}(\vx_1,X_2)\right)$ and 
$\rho_n=\rho_{1n}=an^{-1/4}$, $\rho_{2n}=bn^{-1/4}$, for constants $a$, $b$.

\begin{thm} \label{thm:additive_alternative} 
Consider the notation and assumptions of Theorem \ref{thm:main_X1}. Moreover, assume that 
$\tilde m_2(x)$ is Lipschitz continuous. 

\vspace{-4mm}
 \begin{enumerate}
 \item ({\bf Local Additive Alternatives})
Then, under $H_1^A$ in (\ref{additive_alternative}),  as $n \rightarrow \infty$, 
$$n^{1/2}(MST - MSE) \dkonv N\left(a^2 p \Var(\tilde m_2(X_2)),\frac{2p(2p-1)}{3(p-1)}\tau^2\right).$$
 \item ({\bf Local General Alternatives}) Assume further that $\tilde m_{12}(\vx_1,x_2)$ is Lipschitz continuous on $x_2$ uniformly on $\vx_1$.  Then, under $H_1^G$ in (\ref{general_alternative}), as $n \rightarrow \infty$, 
$$n^{1/2}(MST - MSE) \dkonv N\left(\mu^G,\frac{2p(2p-1)}{3(p-1)}\tau^2\right), \ \mbox{where}$$
$
\mu^G= pa^2Var(\tilde m_2(X_2))\hspace{-1mm} +\hspace{-1mm} pb^2Var(\tilde m_{12}(\vX_1,X_2))\hspace{-1mm} +\hspace{-1mm} 2pabCov(\tilde m_2(X_2),\tilde m_{12}(\vX_1,X_2))$.
If $a = b$ the formula simplifies to
$
\mu^G=pa^2Var(\tilde m_2(X_2)\hspace{-1mm} +\hspace{-1mm} \tilde m_{12}(\vX_1,X_2)).
$
\end{enumerate}
\end{thm}

\vspace{-6mm}

\subsection{Practical considerations}

 \subsubsection{Using other estimators of $m(\vx_1)$}\label{other.est}

\indent \indent We conjecture that the asymptotic theory of the test statistic remains the same for a wide class of other nonparametric estimators of $m_1$, such as local polynomial estimators or, under an additive model, the backfitting estimator. Moreover, if one is willing to assume additional smoothness conditions then, since use local polynomial estimators yields faster rates of convergence (Stone, 1982), we conjecture that the asymptotic theory of the test statistic based on such estimators can include covariate dimensionality greater than the present 4. A similar comment applies if one is willing to assume an additive model.

An alternative version of the present kernel estimator incorporated in our simulations is a version of the estimator proposed by Newey (1994) and 
Linton and Nielsen (1995), and further studied in Mammen, Linton and Nielsen (1999)  and Horowitz and Mammen (2004), computed as 
 \bqa
 \widehat{\tilde m}_1(\vx_1)=\frac1n \sum_{i=1}^n\widehat m(\vx_1,X_{2i}),
 \eqa
with $\widehat m(\vx_1,X_{2i})$ a Nadaraya-Watson kernel estimator of $m(\vx_1,x_2)$. 
Under the null hypothesis this also estimates $E(Y|\vx_1)$, but under the alternative it estimates (see decomposition (\ref{m.decomposition}))
\bqa
\tilde m_1(\vx_1) 
=\mu +\tilde m_1(\vx_1).
\eqa

\renewcommand{\baselinestretch}{1}
\begin{table}[ht]
\caption{Rejection rates with alternative fitting methods}  \label{table:Different_estimators}
\centering     
\begin{tabular}{l c c c c c}  
\hline\hline  
& \multicolumn{1}{c}{$\theta = 0$} & \multicolumn{4}{c}{$\theta = 1$} \\
& \multicolumn{1}{c}{$\gamma$} & \multicolumn{4}{c}{$\gamma$} \\
\cmidrule(r){2-2}  \cmidrule(r){3-6}
Method  & 0 & 1 & 2 & 3 & 4  \\ [0.5ex]  
\hline
 ANOVA-type   (p=11) & .053 & .119 & .196 & .292 & .390 \\
 ANOVA-type   (p=9)  & .052 & .121 & .191 & .296 & .388 \\
 ANOVA-type 2 (p=11) & .054 & .150 & .218 & .315 & .440 \\
 ANOVA-type 2 (p=9)  & .054 & .122 & .201 & .320 & .422  \\
\hline   
\hline    
\end{tabular} 
\end{table}
\renewcommand{\baselinestretch}{1.7}

\noindent In contrast, under the alternative, $\widehat m_1(\vx_1)$ estimates
\bqa
m_1(\vx_1)=\mu+ \tilde m_1(\vx_1) + E(\tilde m_2(X_2) + \tilde m_{12}(\vx_1,X_2)|\vX_1=\vx_1).
 \eqa
Thus, forming the residuals by $\hat\xi=Y-\widehat m_1(\vx_1)$  inadvertently removes some of the effect of $X_2$. The simulations reported in Table \ref{table:Different_estimators} suggest that  the test statistic using the residuals $\hat\xi=Y-\widehat{\tilde m}_1(\vx_1)$ (ANOVA-type2 in the table) can have improved power against non-additive alternatives. In this table, the data are generated according to the model 
$Y = X_1 + \theta X_2 + \gamma X_1 X_2 + \epsilon$, with $X_1, X_2$ independent $U(0,1)$ random variables and $\epsilon \sim N(0,3^2)$. The reported rejection rates are based on 2000 simulation runs with $n=100$. We conjecture that a similar alternative to the local polynomial estimator of $m_1$ will have improved power against non-additive alternatives.

\vspace{-3mm}

 \subsubsection{Using dimension reducing techniques}\label{subsec.dr}
 
\indent \indent The conditions of Theorem \ref{thm:main_X1}  restrict $d$ to be less than or equal to $4$.
However, under the assumption of a sparse model, the effects of the curse of dimensionality can be moderated through the use of dimension reduction techniques. In all simulations reported in Section \ref{sec.sim_var.sel} as well as the data analysis results of Section \ref{sec.real_data} we used the classical sliced inverse regression (SIR) dimension reduction method of Li (1991). Moreover, we employed  a variable screening method prior to applying SIR. The variable screening consists of performing the marginal test of Wang, Akritas and Van Keilegom (2008) for the significance of each variable, and keeping those variables for which the p-value is less than 0.5. 

\section{Simulations: Model Checking Procedures}\label{sec:simulations:model_checking}

\vspace{-3mm}

\subsection{Brief literature review}
\indent \indent Let $\vX=(\vX_1,\vX_2)$ be the vector of $d$ available predictors, with $\vX_1$ being $d_1$-dimensional. The problem of assessing the usefulness of $\vX_2$, i.e. testing $H_0: m(\vx_1,\vx_2)=m_1(\vx_1)$, has been approached from different angles by many authors. The literature is extensive, so only a brief summary of some of the proposed ideas  and the resulting test procedures is given below. For additional references see Hart (1997) and Racine, Hart and Li (2006). 

One class of procedures is based on the idea that the null hypothesis residuals,
$\xi = Y - m_1(\vX_{1})$, satisfy $E(\xi |\vX)=0$ under $H_0$ and $E(\xi |\vX)=
m(\vX)-m_1(\vX_{1})$ under the alternative. Thus, $E(\xi E(\xi |\vX)|\vX)=(m(\vX)-m_1(\vX_{1}))^2$ under the alternative and zero under the null. Using this idea, Fan and Li (1996) propose a test statistic based on estimating
$E[\xi f_1(\vX_1)E(\xi f_1(\vX_1)|\vX)f(\vX)]$ which equals $E[(m(\vX)-m_1(\vX_{1}))^2f_1(\vX)^2f(\vX)]$ under the alternative and zero under the null. Their test statistic is
\begin{equation}
\frac{1}{n}\sum_{i}[\tilde{\xi}_i\tilde{f}_1(\vX_{1i})]\left[\frac{1}{(n-1)h_n^d}\sum_{j\neq i}[\tilde{\xi}_j\tilde{f}_1(\vX_{1j}]K\left(\frac{\textbf{X}_{i} - \textbf{X}_{j}}{h_n}\right)\right] \nonumber 
\end{equation}
where $\tilde{f}_1$ is the estimated density of $\vX_1$, $\tilde{\xi}_i$ is the estimated residuals under the null hypothesis, and $K$ is a kernel function. 
Fan and Li (1996) show that their test statistic is asymptotically normal under $H_0$. 
 Lavergne and Voung (2000) propose a test statistic based on different estimator of the same quantity as Fan and Li (1996), 
 which is 
  \bqa
\frac{(n-4)!}{n!}\sum_{a}(Y_i - Y_k)(Y_j - Y_l)L_n
\left(\frac{\textbf{X}_{1i} - \textbf{X}_{1k}}{g_n}\right)
L_{n}\left(\frac{\textbf{X}_{1j} - \textbf{X}_{1l}}{g_n}\right)
K_{n}\left(\frac{\textbf{X}_{i} - \textbf{X}_{j}}{h_n}\right),
\eqa
where $\sum_{a}$ is the sum over all permutations of 4 distinct elements chosen from $n$, $L_n=g_n^{-d_1}L$ for a kernel $L$ on $\mathbf{R}^{d_1}$ and $K_n=h_n^{-d}K$ for a kernel $K$ on $\mathbf{R}^d$. Lavergne and Voung (2000) show that their test statistic is also asymptotically normal under $H_0$.

A related class of procedures is based on direct estimation of $E[(m(\vX)-m_1(\vX_1))^2W(\vX)]$, for some weight function $W$. See, for example, A\"{i}t-Sahalia, Bickel,  and Stoker (2001). The use of such test statistics is complicated by the need to correct for their bias. See also the bootstrap-based procedure of Delgado and Manteiga (2001). Because of the computer intensive nature of bootstrap-based procedures, these are not included in our comparisons.

An additional class of test procedures uses alternatives based on Stone's (1985) additive model. We will consider the procedure proposed by
Fan and Jiang (2005). This is based on Fan, Zhang, and Zhang's (2001) Generalized Likelihood Ratio Test (GLR), using a local polynomial approximation and the backfitting algorithm for estimating the additive components.  


\vspace{-3mm}

\subsection{Numerical comparison}
 \label{sec:simulations:variable_selection}

\indent \indent In this section we 
compare the proposed ANOVA-type and ANOVA-type2 statistics described in Section \ref{other.est} to the statistics proposed by Lavergne and Vuong (2000) (LV in the tables), Fan and Li (1996) (FL in the tables), and Fan and Jiang (2005) (GLR in the tables). 

The data is generated according to the models (also used in Lavergne and Vuong, 2000)
\begin{equation}
Y = -X_1 + X_1^3 + f_j(X_2) + \epsilon, \mbox{    } j=0,1,2,3,,4,5,6,
\end{equation}
where $X_1, X_2$ are iid N$(0,1)$ and $\epsilon \sim N(0,4)$. Here, $f_0(x)=0$, which corresponds to the null hypothesis $H_0: m(x_1,x_2) = m(x_1)$; $f_1(X_2) = .5 X_2, f_2(X_2) = X_2$ and $f_3(X_2) = 2 X_2$ give three linear alternatives,and $f_4(X_2) = sin(2 \pi X_2), f_5(X_2) = sin(\pi X_2)$, and $f_6(X_2) = sin(2/3 \pi X_2)$ give three non-linear alternatives. The kernel for the
Nadaraya-Watson estimation of $m(X_1)$ is the uniform on $(-0.5,0.5)$ density, and the bandwidth is selected through leave-one-out cross validation. 
The rejection rates shown in Table \ref{table:stat1_LV_FL} for LV, FL, and F tests are taken from the simulation results reported in the LV paper (based on 2000 runs). It is important to note that, in each simulation setting, the LV paper reports several rejection rates for the LV and FL tests, each corresponding to different values of smoothing parameters. Since the best performing constants are different for different simulation settings, the rejection rates reported in Table \ref{table:stat1_LV_FL} represent a) the most accurate alpha level achieved over all constants, and b) the best power achieved overall constants for each alternative. For comparison purposes, the rejection rates for the ANOVA-type tests and the GLR test are also based on 2000 simulation runs. 
\renewcommand{\baselinestretch}{1}
\begin{table}[ht] 
\caption{Rejection rates under $H_0$, linear and non-linear alternatives}  \label{table:stat1_LV_FL}
\centering     
\begin{tabular}{c c c c c c c c c}  
\hline\hline                       
& & \multicolumn{4}{c}{linear} & \multicolumn{3}{c}{sine} \\
\cmidrule(r){3-6}  \cmidrule(r){7-9}
n & test & $f_0$ & $f_1$ & $f_2$ & $f_3$ & $f_4$ & $f_5$ & $f_6$ \\ [0.5ex]  
\hline
100 & LV   & .041 & .098 & .482 & .991 & .182 & .266 & .319   \\ 
    & FL   & .021 & .051 & .271 & .970 & .126 & .168  & .187  \\
    & ANOVA-type (p = 9) & .052 & .218 & .79 & .999 & .423 & .523 & .535 \\      
    & ANOVA-type (p = 7) & .056 & .244 & .780 & .999 & .432 & .527 & .551 \\      
    & ANOVA-type2  (p = 9)   & .065 & .275 & .831 &  1 & .453 & .598 & .600 \\
    & GLRT               & .044 & .365 & .951 & 1    & .123 & .497 & .645 \\
    & F-test & .051 & .695 & .997 & 1 & .046 & .055 & .222 \\                   
\hline
200 & LV   & .054 & .208 & .875 & 1 & .386 & .540 & .678 \\
    & FL   & .025 & .083 & .695 & 1  & .289 & .395 & .471 \\
    & ANOVA-type (p = 9) & .055 & .374 & .95 & .999 & .73 & .778 & .788 \\    
    & ANOVA-type (p = 7) & .051 & .376 & .979 & 1 & .746 & .820 & .820 \\        
    & ANOVA-type2  (p = 9)      & .069 & .487 & .999  & 1 & .821 & .882 & .884 \\ 
    & GLRT               & .036 & .656 & 1    & 1 & .188 & .877 & .936 \\ 
    & F-test & .052 & .931 & 1 & 1 & .051 & .053 & .340 \\                     
\hline    
\hline    
\end{tabular} 
\end{table} 
\renewcommand{\baselinestretch}{1.7}

As expected, the F test achieved the best results for the three linear alternatives and the worse results for the three non-linear alternatives. 
The GLR test has higher power than the ANOVA-type tests against linear alternatives (which is partly explained by the fact it is based on normal likelihood),  but is much less powerful against the first of the non-linear alternatives. As the non-linearity decreases ($f_5$ and $f_6$) the power of the GLR test improves. 

The GLR test is designed for additive models, which is exactly the simulation setting of Table \ref{table:stat1_LV_FL}. Under non-additive alternatives, however, it can perform poorly as indicated by the simulations  reported in the first part of Table \ref{table:GLR_non_add}. These simulations use sample size $n=200$ with data generated from the model $Y = X_1^{X_2}(1 + \theta X_3) + \frac{X_2^{(1 + \theta X_3)}}{X_2} + \epsilon$, where $\epsilon \sim N(0,0.1)$, and $X_1, X_2, X_3$ are i.i.d. $U(0.5,2.5)$. The hypothesis tested is that $m(X_1,X_2,X_3)=m_1(X_1,X_2)$. The residuals for the ANOVA-type test in the first part of Table \ref{table:GLR_non_add} are based on a Nadaraya-Watson fit with kernel the uniform on 
$(-0.5,0.5)\times (-0.5,0.5)$ density and the common bandwidth selected through leave-one-out cross validation.

\renewcommand{\baselinestretch}{1}
\begin{table}[ht]  
\caption{Rejection rates for non-additive and heteroscedastic models}  \label{table:GLR_non_add}
\centering     
\begin{tabular}{l c c c c c c c c c c}  
\hline\hline 
&&&&&\\[-2mm]
& \multicolumn{5}{c}{non-additive alternatives} & \multicolumn{5}{c}{heterocedastic alternatives}\\ 
&  \multicolumn{5}{c}{$\theta$}  & \multicolumn{5}{c}{$\theta$} \\
\cmidrule(r){2-6}  \cmidrule(r){7-11}
test & 0  & 0.02 & .04 & .06 & .08 & 0  & 0.025 & .05 & .1 & .2 \\ [0.5ex]  
\hline
 ANOVA-type ($p=9$)    & .052 & .176 & .609 & .940 & .994   &  .053 & .067 & .124 & .485  & .998\\
 GLR                   & .048 & .082 & .110 & .189 & .304   & .465 & .511 & .624 & .908 & 1\\
\hline   
\hline    
\end{tabular} 
\end{table} 
\renewcommand{\baselinestretch}{1.7}


\noindent Finally, it should be mentioned that the GLR test does not maintain its level under heteroscedasticity. In simulations, reported in the second part of Table \ref{table:GLR_non_add},  under the additive but heteroscedastic model 
$Y = X_1^2 + \theta \cos(\pi X_2) + X_2\epsilon$, $X_1, X_2$ i.i.d. N$(0,1)$, $\epsilon \sim \mbox{N}(0,0.5)$, using sample size $n=200$, the GLR test is very liberal while the ANOVA-type test maintains an accurate level.

\vspace{-6mm}
\section{From Model Checking to Variable Selection} \label{section:Variable_Selection}
\vspace{-2mm}
\subsection{The proposed procedure} \label{section:Benjamini}

\indent \indent In this section we will assume a sparse regression model in the sense that there exits a subset of indices $I_0=\{j_1,\ldots,j_{d_0}\}\subset \{1,\ldots,d\}$ such that only the covariates $X_j$ with $j\in I_0$ influence the regression function. Moreover, we will assume the dimension reduction model of Li (1991), i.e.
$
m(\vx)=g(\vB\vx), \ \mbox{where  $\vB$ is a $K\times d$ matrix}.
$
In this context we will describe the following variable selection procedure using backward elimination based on the Benjamini and Yekuteli (2001) method for controlling the false discovery rate (FDR):
\vspace{-3mm}
\ben
\item Apply the variable screening procedure described in Section \ref{subsec.dr}. 
With a slight abuse of notation, the vector of the remaining covariates and its dimension will be denoted by $\vx$ and $d$.
\vspace{-3mm}
\item Use SIR to obtain the estimator $\widehat\vB$.
\vspace{-3mm}
\item\label{ret.step} Obtain p-values from testing each of the hypotheses:
\begin{equation} \label{H0_benjamini} 
H_{0}^j: m(\textbf{x}) = m_1(\textbf{x}_{(-j)}), \mbox{    } j = 1,\ldots, d,
\end{equation} 
where $\textbf{x}_{(-j)} = (x_1,...,x_{j-1},x_{j+1},...,x_d)$:
    \ben
    \vspace{-3mm}
    \item Compute the test statistic (see Theorem \ref{thm:main_X1}) 
\bqa
z_j = \sqrt{n}(MST_j - MSE_j)/\sqrt{\frac{2p(2p-1)}{3(p-1)}}\hat\tau_j^2
\eqa
using residuals formed by a kernel estimator on the variables $\widehat\vB_{(-j)}\vx_{(-j)}$, where 
$\widehat\vB_{(-j)}$ is the $K\times (d-1)$ matrix obtained by omitting the $j$th column of  $\widehat\vB$. 
   \item\label{ret.step2} Compute the p-value for $H_0^j$ as $\pi_j = 1-\Phi(z_j)$.
   \een
  \item Compute 
  \vspace{-3mm}
  \bqan\label{B-Y-k}
k = \max\left\{ j: \pi_{(j)} \leq \frac{i}{d}\frac{\alpha}{\sum_{l=1}^{d}l^{-1}}\right\}
\eqan
for a choice of level $\alpha$, where $\pi_{(1)},\ldots,\pi_{(d)}$ are the ordered p-values. If $k=d$ stop and retain all variables. If $k<d$ 
  \ben
  \item update $\vx$ by eliminating the covariate corresponding to $\pi_{(d)}$, 
  \item update $\widehat\vB$ by eliminating the column corresponding to the deleted variable, and
  \item proceed to the next step.
  \een 
  \vspace{-3mm}
 \item Repeat steps \ref{ret.step} and \ref{ret.step2}, with the updated $vx$ and $\widehat\vB$.
\een
\vspace{-3mm}
{\bf Remarks.} 1)  Another approach for constructing a variable selection procedure is to use a single application of the Benjamini
and Yekuteli (2001) method for controlling the false discovery rate (FDR). This is
similar to one of the two procedures proposed in Bunea et al. (2006). However, this did not perform well  in simulations and is not recommended. 
A backward elimination approach was Li, Cook and Nachtsheim (2005), but they did not use multiple testing ideas.\\
\indent\indent 2) Based on our simulation results, the variable screening part (Step 1) of the variable selection procedure does not improve the performance. However, it was included in the simulations as it reduces the computational time.

\vspace{-6mm}

\subsection{Simulations: Variable selection procedures}\label{sec.sim_var.sel}
\vspace{-3mm}
\indent \indent Because the ANOVA-type2 method (see Section \ref{other.est}) is computationally more intensive, we used only the proposed variable selection method using the ANOVA-type test described in Section \ref{test.stat}, with ANOVA cell sizes of 5 (when $n=40$), 7, and 9 (when $n=110$). The parameter $\alpha$ was set to $0.07$ in Table \ref{table:LASSO} (so the FDR is controlled at level $(25-5)0.07/25=0.056$), and $\alpha=0.06$ in Table \ref{table:LASSO2_sin} (so the FDR is controlled at levels  0.052 and 0.045). These procedures are compared 
with LASSO, SCAD, adaptive LASSO, the FDR-based variable selection method proposed by Bunea, Wegkamp and Auguste (2005) (BWA in the tables), and a version of the BWA procedure which uses backward elimination (BWA+BE in the tables). 
The comparison criterion is the mean number of correctly and incorrectly excluded variables. All comparisons are based on 2000 simulated data sets.

For LASSO we found that the R code in in http://cran.r-project.org/web/packages/\\ glmnet/index.html, with the lambda.lse option for selecting lambda, gave the best results; for adaptive LASSO we used the R code from
http://www4.stat.ncsu.edu/$\sim$boos/var.select\\
/lasso.adaptive.html; for SCAD we used the function scadglm of the package SIS in R.

In Table \ref{table:LASSO}, data sets of size $n=110$ were generated from the linear model $Y = \vbeta^T\vX + \epsilon,$ where  $\epsilon \sim N(0,3^2)$, the dimension of $\vX$ is $d=25$, and 
\bqa
\vbeta^T = (3,1.5,0,0,2,0,2,0,0,0,0,0,0,0,0,0,3,0,0,0,0,0,0,0,0).
\eqa 
The covariates are generated from a multivariate normal distribution with marginal means zero and covariances as shown in the table. It is seen that the proposed nonparametric variable selection procedures correctly exclude, on average, about 19.5 out of the 20 nonsignificant predictors. This is about as good as the procedures  
designed for linear models. The proposed procedures incorrectly exclude, on average, about 0.5 of the 5 significant predictors, which is more than the other procedures (with the exception of BWA). 

\renewcommand{\baselinestretch}{1}
\begin{table}[ht] 
\caption{Comparisons using a linear model: $d=25$, $n=110$}  \label{table:LASSO}
\centering     
\begin{tabular}{l c c c c}  
\hline\hline  
& \multicolumn{2}{c}{$\Sigma=I$} &  \multicolumn{2}{c}{$\Sigma=(0.5^{|i-j|})$}\\
\cmidrule(r){2-3}  \cmidrule(r){4-5}
test & correct & incorrect   & correct & incorrect   \\ [0.5ex]  
\hline
SCAD & 19.48  & .026 & 19.37   & .023\\ 
LASSO    & 18.29 & .005 & 18.28 & .004\\                  
Adaptive LASSO & 19.28 & .005 & 19.26 & .025\\  
BWA & 19.99 & 1.02 & 19.97 & 1.41\\ 
BWA+BE & 19.55 & .001 & 19.49 & .041\\ 
ANOVA-type(p=7) & 19.46 & .63  & 19.30 & .44\\
ANOVA-type(p=9) & 19.52 & .65  & 19.40 & .36\\
\hline    
\hline    
\end{tabular} 
\end{table}
\renewcommand{\baselinestretch}{1.7}

\indent In Table \ref{table:LASSO2_sin}, data sets of size $n=40$ were generated from the 
models $Y=g_\ell(\vX)+\epsilon$, $\ell=1,2$, where $\epsilon\sim N(0,0.3^2)$, the dimension of $\vX$ is $d=8$, and
\bqa
g_1(\vx)= \sin(\pi x_1)  ,\ \ \ g_2(\vx)=\sin(3/4 \pi x_1)-3\Phi(-|x_5|^3).
\eqa
\vspace{-3mm}
\renewcommand{\baselinestretch}{1}
\begin{table} [htb!]
\caption{Comparisons using nonlinear models: $d=8$, $n=40$}   \label{table:LASSO2_sin}
\centering     
\begin{tabular}{l c c c c c c c}  
\hline\hline  
& \multicolumn{2}{c}{$g_1$} & \multicolumn{2}{c}{$g_2$} \\
\cmidrule(r){2-3}  \cmidrule(r){4-5}
test & correct & incorrect  & correct & incorrect \\ [0.5ex]  
\hline
SCAD &  6.74  & .96 & 5.71 & 1.79\\ 
LASSO    & 6.59  & .92 & 5.72 & 1.80\\                   
Adaptive LASSO & 6.65 & .95 & 5.62 & 1.73\\
BWA & 6.99 & 1 & 5.99 & 1.99\\
BWA+BE & 6.65 & .94 & 5.70 & 1.75\\
ANOVA-type(p=7)& 6.21 & 0.001 & 5.75 & .11\\
ANOVA-type(p=5)& 6.39 & 0.001 & 5.71 & .08\\
\hline   
\hline    
\end{tabular} 
\end{table}
\renewcommand{\baselinestretch}{1.7}
The covariates are generated as normal with marginal means zero and covariance matrix $\Sigma=(0.5^{|i-j|})$. It is seen that the linear model based procedures fail to select the significant predictor(s) almost always. On the other hand, the proposed procedures always select the one relevant predictor under model $g_1$, and exclude incorrectly, on average, about 0.08 out of the two important predictors under model $g_2$.  

\subsection{Real Data Example: Body Fat Dataset}\label{sec.real_data}

\indent \indent The Body Fat data is supplied by Dr. A. Garth Fisher for non-commercial purposes, and it can be found at "http://lib.stat.cmu.edu/datasets/bodyfat". The data set contains 
measurements of percent body fat (using Siri's (1956) method),
 Age (years),
 Weight (lbs),
 Height (inches), circunferences of
 Neck (cm),
 Chest (cm),
 Abdomen (cm),
 Hip (cm),
 Thigh (cm),
 Knee (cm),
 Ankle (cm),
 Biceps (cm),
 Forearm (cm) and
 Wrist (cm), from 252 men. The response variable is the percentage of body fat.\\
\indent  We compare the results of SCAD, LASSO, Adaptive LASSO and BWA with backward elimination to the ANOVA-type procedure with variable screening and SIR, as described in  Section \ref{section:Benjamini}. Table \ref{table:BodyFat} shows the results for LASSO, SCAD, Adaptive LASSO and BWA.
\renewcommand{\baselinestretch}{1}
\begin{table}[ht] \label{table:BodyFat}
\caption{Results for LASSO, Adaptive LASSO, SCAD, BWA}  
\centering     
\begin{tabular}{l c c c c c}  
\hline\hline  
Predictor &  LASSO & Adpt. LASSO & SCAD &  BWA&  \\
 Age    &  .06499 & 0 & .001061    &  0 & \\
 Weight &   0 & -.09511 & -.11688 &  -.1356 & \\ 
 Height &   -.1591 & 0 & -.05818 & 0 & \\ 
 Neck   &   -.2579 & 0 & 0 & 0  & \\
 Chest  &    0 & 0 & 0 &   0& \\
 Abdomen &  .7079 & .9113 & .9052 & .9958   &  \\ 
 Hip    &   0 & 0 & 0 &  0&   \\
 Thigh  &   0 & 0 & 0 &0  &    \\
 Knee   &   0 & 0 & 0 &  0&   \\
 Ankle  &   0 & 0 & 0 &  0&    \\
 Biceps  &  0 & 0 & 0 &  0&   \\
 Forearm &   .21756 & 0 & 0  & .4729 & \\ 
 Wrist  &  -1.5353 & -.9871 & 0 & -1.5056 & \\
\hline
\hline    
\end{tabular} 
\end{table} 
\renewcommand{\baselinestretch}{1.7}
\vspace{-2mm}
It is seen that Weight and Abdomen are selected by all except LASSO. This can be explained by the fact that LASSO does not perform well in the presence of highly correlated variables, which is the case with this data set. 
The Adaptive LASSO and BWA give almost the same results but differ considerably from those of SCAD. 

\indent For the ANOVA-type method we used SIR with the number of slices ranging from 2 to 100.
Abdomen, Weight, Biceps and Knee were selected with 99, 87, 88 and 23 of the 99 different numbers of slices, respectively. All other variables were selected less than 15 
times. On the basis of these results we recommend a model based on Abdomen, Weight and Biceps.

As an explanation of the fact that Biceps was not selected by any of the other methods, we investigated possible violations of the modeling assumptions on which they are based. Marginal plots of the response versus each of the important variables reveal both hetoroscedasticity and nonlinearity.  Moreover, the 99 applications of SIR yielded more than one linear combination (i.e. $K>1$) 50 times. To put this number into perspective, we generated a single set of responses, using the same covariate values with coefficients those from Adaptive LASSO and normal errors using the residual variance. Application of SIR with the number of slices ranging from 2 to 100 on this data set yielded $K=1$ 92 out of the 99 times. This casts serious doubts on the validity of the assumption of a linear model.
 
\vspace{.5cm}
\appendix

\noindent{\huge{{\bf Appendix}}}

\vspace{-8mm}

\section{Auxiliary Results}

\begin{lemma} \label{lemma:Fn} 
Let $X_1,\ldots,X_n$ be iid[$F$], and let $\hat{F}_n(x)$ be the corresponding empirical distribution function. Then, for any constant $c$, $$sup_{x_i,x_j}\left\{|F(x_i) - F(x_j)|I\left[|\hat{F}(x_i) - \hat{F}(x_j)| \leq \frac{c}{n}\right]\right\} = O_p\left(\frac{1}{\sqrt{n}}\right).$$
\end{lemma}
\begin{proof}  By the Dvoretzky, Kiefer and Wolfowitz (1956) theorem, we have that 
$\forall \epsilon \geq 0,$
$$
 P\left(\sup_{x}|\hat{F}_n(x) - F(x)| \geq \epsilon\right) \leq Ce^{-2n\epsilon^2}.
$$
Therefore, $|\hat{F}(x) - F(x)| = O_p\left(\frac{1}{\sqrt{n}}\right)$ uniformly on $x$. Hence, writing
\begin{eqnarray}
|F(x_i) - F(x_j)| &=& |F(x_i) -\hat{F}_n(x_i) + \hat{F}_n(x_i)- F(x_j) + \hat{F}_n(x_j) - \hat{F}_n(x_j)|, \nonumber 
\end{eqnarray}
it follows that $sup_{x_i,x_j}\left\{|F(x_i) - F(x_j)|I\left[|\hat{F}(x_i) - \hat{F}(x_j)| \leq c/n\right]\right\}$ is less than or equal to
\bqa
 &  & sup_{x_i,x_j}\left\{|F(x_i) -\hat{F}_n(x_i)| + |\hat{F}_n(x_j) - F(x_j)|\right\} 
 \nonumber \\
 &&+
  sup_{x_i,x_j}\left\{|\hat{F}_n(x_i)- \hat{F}_n(x_j)|\right\}I\left[|\hat{F}_n(x_i) - \hat{F}_n(x_j)| \leq \frac{c}{n}\right] \nonumber \\
 &=&  O_p\left(\frac{1}{\sqrt{n}}\right) + O_p\left(\frac{1}{\sqrt{n}}\right) + O_p\left(\frac{1}{n}\right).  
\eqa
This completes the proof of the lemma.\end{proof}

\begin{lemma} \label{lemma:Op} 
With $W_i$ be defined in (\ref{def.Wi}), and any Lipschitz continuous function $g(x)$,  
\begin{equation} \label{eqn:Op}
\frac{1}{p}\sum_{j=1}^{n}g(x_{2j})I(j \in W_i) - g(x_{2i}) = O_p\left(\frac{1}{\sqrt{n}}\right), \nonumber
\end{equation}
uniformly in $i=1,\ldots,n$.
\end{lemma}
\begin{proof} First note that by the Lipschitz continuity and the Mean Value Theorem we have
\begin{eqnarray} \label{eqn:Op2}
&& |g(x_{2j}) - g(x_{2i})| \leq M|x_{2j} - x_{2i}| \leq M|F_{X_2}(x_{2j}) - 
F_{X_2}(x_{2i})|/f_{X_2}(\tilde{x}_{ij}), \nonumber
\end{eqnarray} 
for some constant $M$,
where $\tilde x_{ij}$ is between $x_{2j}$ and $x_{2i}$. Thus, 
\bqa
&&\hspace{-4mm}\left |\frac{1}{p}\sum_{j=1}^{n}g(x_{2j})I(j \in W_i) - g(x_{2i})\right| 
\leq \frac{1}{p}\sum_{j=1}^{n}|g(x_{2j}) - g(x_{2i})|I\left[|\hat{F}_{X_2}(x_{2i}) - \hat{F}_{X_2}(x_{2j})| \leq \frac{p-1}{2n}\right] \nonumber \\
  & &\leq  \frac{M}{p}\sum_{j=1}^{n}\frac{|F_{X_2}(x_{2j}) - F_{X_2}(x_{2i})|}{f_{X_2}(\tilde{x}_{ij})}I\left[|\hat{F}_{X_2}(x_{2i}) - \hat{F}_{X_2}(x_{2j})| \leq \frac{p-1}{2n}\right] = O_p\left(\frac{1}{\sqrt{n}}\right), \nonumber
\eqa
where the last equality follows from Lemma \ref{lemma:Fn} and the assumption that $f_{X_2}$ remains bounded away from zero.
\end{proof}
\vspace{-3mm}
As in Wang, Akritas and Van Keilegom (2008), MST-MSE given in (\ref{rel.proposedTS}) can be written as a quadratic form $\hat\vxi_V'A\hat\vxi_V$, where
\bqan\label{def.matrixA}
A = \frac{np-1}{n(n-1)p(p-1)}\oplus_{i=1}^{n}\bold{J}_{p}-\frac{1}{n(n-1)p}\bold{J}_{np}-\frac{1}{n(p-1)}\bold{I}_{np},
\eqan
where $\bold{I}_d$ is a identity matrix of dimension d, $\bold{J}_d$ is a dxd matrix of 1's and $\oplus$ is the Kronecker sum or direct sum. Using arguments similar to those used in the proof of Lemma 3.1 in Wang, Akritas and Van Keilegom (2008), it can be shown that if  $\sigma^2(.,x_2)$, defined in Theorem \ref{thm:main_X1},  is Lipschitz continuous and 
$E(\epsilon_i^4)<\infty$ then, under $H_0$ and as n $\rightarrow \infty$,
\bqan\label{lemma:Ad} 
n^{1/2}[\vxi_V'A\vxi_V - \vxi_V'A_d\vxi_V]\pkonv 0,
\eqan
where $A_d = diag\{B_1,...,B_n\}$, with $B_i = \frac{1}{n(p-1)}[\bold{J}_p - \bold{I}_p].$

\begin{lemma}\label{lem.bandwidth}  For a symmetric, positive definite bandwidth matrix $H_n$, define the norm $||H_n||$ to be the maximum of its eigenvalues. Then we have 
\bqa
\sum_{j_2=1}^{n}w(\textbf{X}_{1i},\textbf{X}_{1j_2})||\textbf{X}_{1j_2} - \textbf{X}_{1i}|| = O(||H_n^{1/2}||).
\eqa
\end{lemma}
\vspace{-4mm}
\begin{proof}  Let  $b$ be such that 
$K(\vx) = K(\vx)I(||\vx||\le \sqrt{d-1}b)$. Such a $b$ exists by the assumption that the density $K$ has bounded support. Thus,
\bqa
&&K\left(H_n^{-1/2}(\textbf{X}_{1i} - \textbf{X}_{1j})\right)||\textbf{X}_{1j} - \textbf{X}_{1i}||  = K\left(H_n^{-1/2}(\textbf{X}_{1i} - \textbf{X}_{1j})\right)||H_n^{1/2}H_n^{-1/2}(vX_{1i}-\vX_{1j})||\\
&\le&K\left(H_n^{-1/2}(\textbf{X}_{1i} - \textbf{X}_{1j})\right)||H_n^{1/2}|| ||H_n^{-1/2}(vX_{1i}-\vX_{1j})||\\
&\le&K\left(H_n^{-1/2}(\textbf{X}_{1i} - \textbf{X}_{1j})\right)||H_n^{1/2}||\sqrt{d-1}b.
\eqa
The statement of the lemma follows from the above.
\end{proof}

\vspace{-1cm}
\section{Proofs of Theorems}
\vspace{-3mm}
\begin{proof}[\textbf{Proof of Theorem \ref{thm:main_X1}}]
\indent Under $H_0$ in (\ref{rel.H0}) we write
 \bqa
 \hat{\xi}_i &=& Y_i - \hat{m}_1(\textbf{X}_{1i}) + m_1(\textbf{X}_{1i}) - m_1(\textbf{X}_{1i}) = \xi_i - (\hat{m}_1(\textbf{X}_{1i}) - m_1(\textbf{X}_{1i})) \\
 &=& \xi_i - \Delta_{m_1}(\textbf{X}_{1i}),
 \eqa
  where $\Delta_{m_1}(\textbf{X}_{1i})$ is defined implicitly in the above relation. Thus, $\hat\vxi_V$ of relation (\ref{hat.vxi}) is decomposed as
$
  \hat\vxi_V=\vxi_V - \vDelta_{m_1V},
$
  where $\vxi_V$ and $\vDelta_{m_1V}$ are defined as in (\ref{hat.vxi}) but using 
  $\xi_i$ and $\Delta_{m_1}(\textbf{X}_{1i})$, respectively, instead of $\hat\xi_i$. Thus $\sqrt{n}(\mbox{MST - MSE})$ can be written as 
  \bqan\label{rel.decom.T}
 \sqrt{n} \hat \vxi_V'A\hat\vxi_V = \sqrt{n}\vxi_V'A\vxi_V -\sqrt{n} 2\vxi_V'A\vDelta_{m_1V} +\sqrt{n} \vDelta_{m_1V}'A\vDelta_{m_1V},
  \eqan
  where the matrix $A$ is defined in (\ref{def.matrixA}).
 The asymptotic normality of $\sqrt{n}\vxi_V'A\vxi_V$ follows by arguments similar to those used in Theorem 3.2 of Wang, Akritas and VanKeilegom (2008). It remains to derive its asymptotic variance and to show that the other two terms  in (\ref{rel.decom.T}) converge to zero in probability.
 Using (\ref{lemma:Ad}) it suffices to find the asymptotic variance of 
$ \sqrt{n}\vxi_V'A_d\vxi_V$. Since $E(\vxi_V'A_d\vxi_V)=0$ its variance equals
$E[(\sqrt{n}\vxi_V'A_d\vxi_V)^2]$. To find this we first evaluate its conditional expectation, 
$E[(\sqrt{n}\vxi_V'A_d\vxi_V)^2| \{X_{2j}\}_{j=1}^n]$, given $X_{21},\ldots,X_{2n}$. 
Recalling the notation
$\sigma^2(.,x_{2}) = E(\xi^2|X_2=x_2)$, we have
 \vspace{-.1mm}
\bqan\label{rel.Pf.Thm1.bounded}
&& \frac{1}{n(p-1)^2}\sum_{i_1,i_2}^{n}\sum_{j_1 \neq l_1}^{n}\sum_{j_2 \neq l_2}^{n}E(\xi_{j_1}\xi_{l_1}\xi_{j_2}\xi_{l_2}|\{X_{2j}\}_{j=1}^n)I(j_s \in W_{i_s},l_s \in W_{i_s},s=1,2)
\nonumber\\[2mm]
&& = \frac{2}{n(p-1)^2}\sum_{i_1=1}^{n}\sum_{i_2=1}^{n}\sum_{j \neq l}^{n}\sigma^2(.,x_{2j})\sigma^2(.,x_{2l})I(j,l \in W_{i_1}\cap W_{i_2})  \\
&& = \frac{2}{n(p-1)^2}\sum_{i_1=1}^{n}\sum_{i_2=1}^{n}\sum_{j \neq l}^{n}\sigma^2(.,x_{2j})\left(\sigma^2(.,x_{2j}) + O_p\left(\frac{p}{\sqrt{n}}\right)\right)I(j,l \in W_{i_1}\cap W_{i_2}) \nonumber \\
&& = \frac{2}{n(p-1)^2}\sum_{j=1}^{n}\sigma^4(.,x_{2j})\sum_{i_1=1}^{n}\sum_{i_2=1}^{n}\sum_{l \neq j}^{n}I(j,l \in W_{i_1}\cap W_{i_2}) + O_p\left(\frac{p^2}{n^{1/2}}\right) \nonumber \\
&& = \frac{2}{n(p-1)^2}\sum_{j=1}^{n}\sigma^4(.,x_{2j})2(1+2^2+3^2+...+(p-1)^2) + O_p\left(\frac{p^2}{n^{1/2}}\right) \nonumber \\
&& = \frac{2}{n(p-1)^2}\frac{p(p-1)(2p-1)}{3}\sum_{j=1}^{n}\sigma^4(.,x_{2j}) + O_p\left(\frac{p^2}{n^{1/2}}\right), \nonumber
\eqan
where the third equality follows from Lemma \ref{lemma:Op} using the assumption that $\sigma^2(.,x_2)$ is Lipschitz continuous and the second last inequality results from the fact that if $1 \leq |j_1-j_2| = s \leq p-1$, then they are $(p-s)^2$ pairs of windows whose intersection includes $j_1$ and $j_2$. Taking limits as $n\to\infty$ it is seen that
 \vspace{-1mm}
\begin{eqnarray}
E\left(n^{1/2}\vxi_V'A_d\vxi_V\right|X_2 = x_2)^2 &\askonv & \frac{2(2p-1)}{3(p-1)}E(\sigma^4(.,X_2)) 
= \frac{2(2p-1)}{3(p-1)}\tau^2.
\end{eqnarray}
From relation (\ref{rel.Pf.Thm1.bounded}) it is easily seen that 
$E[(\sqrt{n}\vxi_V'A_d\vxi_V)^2| \{X_{2j}\}_{j=1}^n]$ remains bounded, and thus
$\Var(n^{1/2}\xi_V'A\xi_V )$ also converges to the same limit by the Dominated Convergence Theorem.
Hence, $n^{1/2}\xi_V'A\xi_V $ converges in distribution to the designated normal distribution. That the second and third terms in (\ref{rel.decom.T}) converge in probability to zero are shown in Lemmas \ref{lem.pf.2nd.A.14}, \ref{lem.pf.3rd.A.14}  , respectively.
\end{proof}

\begin{proof}[\textbf{Proof of Theorem \ref{thm:additive_alternative}.}] 

\subsubsection*{Part 1: Local Additive Alternatives}
 
\indent Note that we can write $\hat{\xi}_j=Y_j - \hat{m}_1(\vX_{1j})$ as
\bqan\label{rel.pf.repr.hatxi}
\hat{\xi}_j 
&=& Y_j - m_1(\vX_{1j}) - \rho_n \tilde m_2(X_{2j}) - [\hat{m}_1(\vX_{1j})- m_1(\vX_{1j})] + \rho_n \tilde m_2(X_{2j}) \nonumber \\
&=& \xi_j - \Delta_{m_1}(\vX_{1j}) + \rho_n \tilde m_2(X_{2j}),
\eqan
and therefore 
$
\hat\vxi_V = \vxi_V - \vDelta_{m_1V} + \rho_n \tilde \vm_{2V}$,
where $\vxi_V$, $\vDelta_{m_1V}$ and $\tilde \vm_{2V}$ are defined as in (\ref{hat.vxi}) but using 
  $\xi_i$, $\Delta_{m_1}(\textbf{X}_{1i})$ and $\tilde m_{2}(X_{2i})$, respectively, instead of $\hat\xi_i$.
Thus, we can write 
\bqan\label{additive.terms}
\sqrt{n}(MST - MSE)&=&
\sqrt{n}\hat\vxi_V'A\hat\vxi_V = \sqrt{n}(\vxi_V - \vDelta_{m_1V})'A(\vxi_V - \vDelta_{m_1V}) + \nonumber \\
&+& \sqrt{n}2\rho_n (\vxi_V - \vDelta_{m_1V})'A\tilde \vm_{2V} + \sqrt{n}\rho_n^2 \tilde \vm_{2V}'A\tilde \vm_{2V}. 
\end{eqnarray}
By Theorem \ref{thm:main_X1}, $\sqrt{n}(\vxi_V - \vDelta_{m_1V})'A(\vxi_V - \vDelta_{m_1V})\dkonv N(0, [2p(2p-1)\tau^2]/[3(p-1)])$.
That $\sqrt{n}2\rho_n (\vxi_V - \vDelta_{m_1V})'A\tilde\vm_{2V} \pkonv 0$ and $\sqrt{n}\rho_n^2 \tilde \vm_{2V}'A\tilde\vm_{2V} \pkonv a^2pV(\tilde m_2(X_2))$ are shown in Lemma \ref{lem.pf.Add.Second} and Lemma \ref{lem.pf.Add.Third}, respectively. This completes the proof of part 1.

\noindent {\bf Part 2: Local General Alternatives}

\noindent Working as in (\ref{rel.pf.repr.hatxi}) we can write 
$
\hat \vxi_V = \vxi_V - \vDelta_{m_1V} + \rho_{1n}\tilde \vm_{2V} + \rho_{2n}\tilde \vm_{12V}$,
where $\vxi_V$, $\vDelta_{m_1V}$, $\tilde \vm_{2V}$ and $\tilde \vm_{12V}$ are defined as in (\ref{hat.vxi}) but using 
  $\xi_i$, $\Delta_{m_1}(\textbf{X}_{1i})$, $\tilde m_{2}(X_{2i})$ and $\tilde m_{12}(\vX_{1i},X_{2i})$, respectively, instead of $\hat\xi_i$.
Thus $\sqrt{n}(MST - MSE)$ is
\begin{eqnarray} \label{general.terms}
 \sqrt{n}\hat \vxi_V'A\hat \vxi_V 
&=& \sqrt{n}(\vxi_V - \vDelta_{m_1V} - \rho_{1n}\tilde \vm_{2V})'A(\vxi_V - \vDelta_{m_1V} - \rho_{1n}\tilde \vm_{2V}) \nonumber \\
&+& \sqrt{n}2\rho_{2n} (\vxi_V - \vDelta_{m_1V} - \rho_{1n}\tilde \vm_{2V})'A\tilde \vm_{12V} + \sqrt{n}\rho_{2n}^2 \tilde \vm_{12V}'A\tilde \vm_{12V}. 
\end{eqnarray}
By Part 1 of the theorem, 
$\sqrt{n}(\vxi_V - \vDelta_{m_1V} - \rho_{1n}\tilde \vm_{2V})'A(\vxi_V - \vDelta_{m_1V} - \rho_{1n}\tilde \vm_{2V})$ converges in distribution to
\bqa
N(a^2pVar(m_2(X_2)), [2p(2p-1)\tau^2]/[3(p-1)]).
\eqa
Hence, it is enough to show that $\sqrt{n}\rho_{2n}^2 \tilde \vm_{12V}'A\tilde \vm_{12V} \pkonv pb^2Var(\tilde m_{12}(\vX_1,X_2))$ and 
$\sqrt{n}2\rho_{2n} (\vxi_V - \vDelta_{m_1V} - \rho_{1n}\tilde \vm_{2V})'A\tilde \vm_{12V} \pkonv 2pabCov(\tilde m_2(X_2),\tilde m_{12}(\vX_1,X_2))$. These are shown in Lemmas \ref{pf.General.Third} and \ref{pf.General.Second}, respectively.
\end{proof}

\vspace{-4mm}

\section{Some Detailed Derivations}

\begin{lemma}\label{lem.pf.2nd.A.14}
The second term in (\ref{rel.decom.T}) converges in probability to zero, i.e.
 \bqa
 T_{2n}:=\sqrt{n}\vxi_V'A\vDelta_{m_1V}\pkonv 0.
 \eqa
\end{lemma}
\begin{proof} After some algebra it can be seen that
\bqan  \label{rel.1st.term}
\hspace{-0.8cm}T_{2n} &=& \frac{n^{-1/2}(np -1)}{(n-1)p(p-1)}\sum_{i=1}^{n}\sum_{j\in W_i}\xi_j\sum_{k\in W_i}\Delta_{m_1}(\textbf{X}_{1k}) \nonumber \\
 &&- \frac{n^{-1/2}p}{(n-1)}\sum_{i=1}^{n}\xi_i\sum_{j=1}^{n}\Delta_{m_1}(\textbf{X}_{1j}) -\ \frac{n^{-1/2}p}{(p-1)}\sum_{i=1}^{n}\xi_i\Delta_{m_1}(\textbf{X}_{1i}). 
\eqan
We will show that each of the three terms above converge in probability to zero conditionally on the set of observed predictors, $\{\vX_j\}_{j=1}^n$, and thus also unconditionally. Note that, because all windows $W_i$ are of finite size $p$, 
the first term on the 
right hand side of (\ref{rel.1st.term}) can be written as a finite sum of $p^2$ terms each of which is similar to the last term in (\ref{rel.1st.term}). Thus, it suffices to show that the last and second terms of (\ref{rel.1st.term}) converge to zero.
For notational simplicity, all expectations and variances in this proof are to be understood as conditional on $\{\vX_j\}_{j=1}^n$.  
For the last term in (\ref{rel.1st.term}) we have 
\begin{eqnarray} \label{eqn:prod}
n^{-1/2}\sum_{i=1}^{n}\xi_i\Delta_{m_1}(\textbf{X}_{1i}) 
&=& n^{-1/2}\sum_{i=1}^{n}\sum_{j=1}^{n}w(\textbf{X}_{1i},\textbf{X}_{1j})(m_1(\textbf{X}_{1j}) +\xi_j - m_1(\textbf{X}_{1i}))\xi_i \nonumber \\
&=& n^{-1/2}\sum_{i=1}^{n}\sum_{j=1}^{n}w(\textbf{X}_{1i},\textbf{X}_{1j})(m_1(\textbf{X}_{1j}) - m_1(\textbf{X}_{1i}))\xi_i \nonumber \\ 
&& +\ n^{-1/2}\sum_{i=1}^{n}\sum_{j=1}^{n}w(\textbf{X}_{1i},\textbf{X}_{1j})\xi_j\xi_i. 
\end{eqnarray}
The first term of the right hand side of (\ref{eqn:prod}) has zero expectation, so it suffices to show that its variance goes to zero. To this end, we write
\begin{eqnarray} 
&& \Var(\frac{1}{\sqrt{n}}\sum_{i=1}^{n}\sum_{j=1}^{n}\xi_i w(\textbf{X}_{1i},\textbf{X}_{1j})(m_1(\textbf{X}_{1j}) - m_1(\textbf{X}_{1i})))  \nonumber \\
&& = \frac{1}{n}\sum_{i=1}^{n}\sum_{j_1=1}^{n}\sum_{j_2=1}^{n}w(\textbf{X}_{1i},\textbf{X}_{1j_1})w(\textbf{X}_{1i},\textbf{X}_{1j_2})\mathsf{x} \nonumber 
\\
&& \hspace{3.5cm} \mathsf{x}(m_1(\textbf{X}_{1j_1}) - m_1(\textbf{X}_{1i}))(m_1(\textbf{X}_{1j_2}) - m_1(\textbf{X}_{1i}))Var(\xi_i) \nonumber \\ 
&& \leq \frac{M}{n}\sum_{i=1}^{n}\sum_{j_1=1}^{n}\sum_{j_2=1}^{n}w(\textbf{X}_{1i},\textbf{X}_{1j_1})w(\textbf{X}_{1i},\textbf{X}_{1j_2})\mathsf{x} \nonumber \\
&& \;\;\;\;\;\;\;\;\;\;\;\;\;\;\;\;\;\;\;\;\;\;\;\;\;\;\;\;\;\; \mathsf{x}(c||\textbf{X}_{1j_1}- \textbf{X}_{1i}||c||\textbf{X}_{1j_2}- \textbf{X}_{1i}||) \nonumber \\ 
&& = Mc^2O(||H_n^{1/2}||)O(||H_n^{1/2}||)  \nonumber, 
 \end{eqnarray}
for some constants $M$ and $c$, where the inequality holds by the assumed conditions for $m_1(\cdot)$, and the last equality follows from Lemma \ref{lem.bandwidth}. Thus, by the assumptions of Theorem \ref{thm:main_X1}  the first term of the right hand side of (\ref{eqn:prod}) goes in probability to zero. To show that the second term in (\ref{eqn:prod}) also goes to 0 in probability since, we will show that its second moment goes to zero. To this end, we write

\begin{eqnarray} \label{3sums}
&& E\left[\frac{1}{n}\sum_{i_1=1}^{n}\sum_{i_2=1}^{n}\sum_{j_1=1}^{n}\sum_{j_2=1}^{n}\xi_{i_1}\xi_{i_2}\xi_{j_1}\xi_{j_2} w(\textbf{X}_{1i_1},\textbf{X}_{1j_1})w(\textbf{X}_{1i_2},\textbf{X}_{1j_2})\right] \nonumber \\
&=&  \frac{1}{n}\sum_{i=1}^{n}\sum_{j=1}^{n}E(\xi_i^2\xi_j^2)\left[w(\textbf{X}_{1i},\textbf{X}_{1j})^2 +w(\textbf{X}_{1i},\textbf{X}_{1i})w(\textbf{X}_{1j},\textbf{X}_{1j})\right.\nonumber\\
&&\hspace{3.8cm}+\ \left. w(\textbf{X}_{1i},\textbf{X}_{1j})w(\textbf{X}_{1j},\textbf{X}_{1i}) \right]\nonumber \\
&& + \frac{1}{n}\sum_{i=1}^{n}E(\xi_i^4)w(\textbf{X}_{1i},\textbf{X}_{1i})w(\textbf{X}_{1i},\textbf{X}_{1i}) \nonumber \\
& \leq & \frac{M_1^2}{n}\sum_{i=1}^{n}\sum_{j=1}^{n}
\frac{c^2}{n^2|H_n|\hat f_{\vX_1}(\vX_{1i})}
\left[\frac{1}{\hat f_{\vX_1}(\vX_{1i})} +\frac{2}{\hat f_{\vX_1}(\vX_{1j})}
\right]\nonumber\\
&& + \frac{M_2}{n}\sum_{i=1}^{n}\frac{c^2}{n^2|H_n|\hat{f}_{\vX_1}(\textbf{X}_{1i})^2}
= O\left(\frac{1}{n|H_n|}\right) + O\left(\frac{1}{n^2|H_n|}\right) , 
\end{eqnarray}
for some constants $M_1,\ M_2$ and $c$, by the fact that $\hat{f}_1$ converges uniformly to $f$ a.s. in the compact 
support $S_{\textbf{X}_1}$ (Ruschendorf 1977). Thus, by the assumptions of Theorem \ref{thm:main_X1}  the second term of the right hand side of (\ref{eqn:prod}) goes in probability to zero.

\indent Consider now the second term in (\ref{rel.1st.term}). Since $n^{-1/2}
\sum_{i=1}^n\xi_i$ remains bounded in probability, its convergence to zero will follow if we show that $n^{-1}\sum_{k=1}^{n}\Delta_{m_1}(\textbf{X}_{1k})\pkonv 0$. For later use, we will actually show that  
 \bqan\label{rel.pf.2ndterm.rel.1st.term}
 \frac{1}{n^{3/4}}\sum_{k=1}^{n}\Delta_{m_1}(\textbf{X}_{1k})=
 \frac{1}{n^{3/4}}\sum_{k=1}^{n}(\hat{m}_1(\textbf{X}_{1k}) - m_1(\textbf{X}_{1k})) 
 \pkonv 0. 
 \eqan
For this we use (cf. Hansen, 2008) 
\bqan \label{kernel.regression.unif.conv.rate}
\sup_{\vx}|\hat{m}_1(\vx) - m_1(\vx)| = O_p(a_n)\ \mbox{ where }\ a_n = \left(\frac{\log n}{n
 \lambda ^{d-1}}\right)^{1/2} + \lambda^2,
\eqan
where $\lambda\to 0$ at the same rate as the eigenvalues $\lambda_i$, $i=1,\ldots,d-1$, of 
$H_n$.
Therefore, the term in the left hand side of (\ref{rel.pf.2ndterm.rel.1st.term}) is of order 
\bqa
\frac{1}{n^{3/4}}n O_p\left(\left(\frac{\log n}{n \lambda ^{d-1}}\right)^{1/2} + \lambda^2\right) = o_p(1),
\eqa
by the assumed conditions stated in (\ref{cond.d_le_4}). This completes the proof of Lemma \ref{lem.pf.2nd.A.14}.
\end{proof}

\begin{lemma}\label{lem.pf.3rd.A.14} The third term in (\ref{rel.decom.T}) converges in probability to zero, i.e.
 \bqa
 T_{3n}=\sqrt{n} \vDelta_{m_1V}'A\vDelta_{m_1V}\pkonv 0.
 \eqa
\end{lemma}
\begin{proof} In this proof we will use $w_{ij}$ to denote $w(\vX_{1i},\vX_{1j})$. Writting
\begin{eqnarray}\label{rel.pf.Delta1}
T_{3n} & =& \frac{\sqrt{n}(np-1)}{n(n-1)p(p-1)}\sum_{i=1}^{n}\left(\sum_{j\in W_i}\Delta_{m_1}(\textbf{X}_{1j})\right)^2 \nonumber \\
&-&\frac{\sqrt{n}p}{n(n-1)}\left(\sum_{i=1}^{n}\Delta_{m_1}(\textbf{X}_{1i})\right)^2 - \frac{\sqrt{n}p}{n(p-1)}\sum_{i=1}^{n}\Delta_{m_1}^2(\textbf{X}_{1i}) ,
\end{eqnarray}
we have to show that each of the three terms on the right hand side of (\ref{rel.pf.Delta1}) converges to zero in probability. First notice that, because all windows $W_i$ are of finite size $p$, the first term on the right hand side of (\ref{rel.pf.Delta1}) can be written as a finite  sum of $p^2$ terms each of which is similar to the last term in (\ref{rel.pf.Delta1}). Therefore, to show that the first and third terms in (\ref{rel.pf.Delta1}) go to zero in probability it is enough to show that $n^{-1/2}\sum_{i=1}^{n}\Delta_{m_1}^2(\textbf{X}_{1i}) \pkonv 0$. Using (\ref{kernel.regression.unif.conv.rate}), it is easy to see that
\bqa
n^{-1/2}\sum_{i=1}^{n}\Delta_{m_1}^2(\textbf{X}_{1i}) \leq n^{1/2}O_p(a_n)^2 = o_p(1).
\eqa
That the second term on the right hand side of (\ref{rel.pf.Delta1}) converges in probability to zero follows directly from  (\ref{rel.pf.2ndterm.rel.1st.term}). 
\end{proof}

\begin{lemma} \label{lem.pf.Add.Second} The second term in (\ref{additive.terms}) converges in probability to zero, i.e.
 \bqa
 \sqrt{n}2\rho_n (\vxi_V - \vDelta_{m_1V})'A\tilde \vm_{2V} \pkonv 0.
 \eqa
\end{lemma}
\begin{proof} By the definition of the matrix $A$, we can write $(\vxi_V - \vDelta_{m_1V})'A\tilde \vm_{2V}$ as
\begin{eqnarray}
&& \frac{np-1}{n(n-1)p(p-1)}\sum_{i=1}^{n}\left[\sum_{j=1}^{n}\tilde m_2(X_{2j})I(j \in W_i)\right]\left[\sum_{k=1}^{n}(\xi_k - \Delta_{m_1}(\vX_{1k}))I(k \in W_i)\right] \nonumber \\
&& -\frac{1}{n(n-1)p}\left[p\sum_{i=1}^{n}\tilde m_2(X_{2i})\right]\left[p\sum_{i=1}^{n}(\xi_i - \Delta_{m_1}(\vX_{1i}))\right] \nonumber \\
&& - \frac{p}{n(p-1)}\sum_{i=1}^{n}\tilde m_2(X_{2i})(\xi_i - \Delta_{m_1}(\vX_{1i})). \nonumber
\end{eqnarray}
Using Lemma \ref{lemma:Op} and the fact that $\tilde m_2(\cdot)$ is Lipschitz continuous, the sum in the first term can be expressed as
\begin{eqnarray}
 && p\sum_{i=1}^{n}[\tilde m_2(X_{2i}) + O(n^{-1/2})]\left[\sum_{k=1}^{n}(\xi_k - \Delta_{m_1}(\vX_{1k}))I(k \in W_i)\right] \leq \nonumber \\
&\leq & p\sum_{k=1}^{n}\left[\sum_{i=1}^{n}\tilde m_2(X_{2i})I(i \in W_k)\right](\xi_k - \Delta_{m_1}(\vX_{1k})) + p^2O(n^{-1/2})\sum_{k=1}^{n}|(\xi_k - \Delta_{m_1}(\vX_{1k}))| \nonumber \\
&=& p^2\sum_{k=1}^{n}\tilde m_2(X_{2k})(\xi_k - \Delta_{m_1}(\vX_{1k})) + O_p(p^2n^{1/2}), \nonumber
\end{eqnarray}
so that
\begin{eqnarray}
\sqrt{n} \rho_n \tilde \vm_{2V}'A(\vxi_V - \vDelta_{m_1V}) &=& \frac{a n^{1.25}p}{n-1}\left[\frac{1}{n}\sum_{i=1}^{n}\tilde m_2(X_{2i})(\xi_i - \Delta_{m_1}(\vX_{1i}))\right] \nonumber \\
 &&\hspace{-1.5cm} -\frac{a n^{1.25}p}{n-1}\left[\frac{1}{n}\sum_{i=1}^{n}\tilde m_2(X_{2i})\right]\left[\frac{1}{n}\sum_{i=1}^{n}(\xi_i - \Delta_{m_1}(\vX_{1i}))\right] + O_p\left(\frac{1}{n^{1/4}}\right). \nonumber 
\end{eqnarray}
Using the fact that $E\left(\tilde m_2(X_{2i})\right) = E\left(\tilde m_2(X_{2i})\xi_i\right) = E(\xi_i) = 0$, relation (\ref{rel.pf.2ndterm.rel.1st.term}) and also that $n^{-3/4}\sum_{i=1}^{n}\tilde m_2(X_{2i})\Delta_{m_1}(\vX_{1i}) \pkonv 0$, as is shown in a similar way to (\ref{rel.pf.2ndterm.rel.1st.term}), completes the proof of the lemma.
\end{proof}

\begin{lemma} \label{lem.pf.Add.Third} The third term in (\ref{additive.terms}) converges in probability to $a^2pV(\tilde m_2(X_2))$, i.e.
 \bqa
 \sqrt{n}\rho_n^2 \tilde\vm_{2V}'A\tilde\vm_{2V} \pkonv a^2pV(\tilde m_2(X_2)).
 \eqa
\end{lemma}
\vspace{-6mm}
\begin{proof} Writing
\begin{eqnarray}
\tilde m_{2V}'A\tilde m_{2V} &=& \frac{np}{n-1}\left\{\left[\frac{1}{n}\sum_{i=1}^{n}\tilde m_2^2(X_{2i})\right] - \left[\frac{1}{n}\sum_{i=1}^{n}\tilde m_2(X_{2i})\right]^2\right\} + O\left(\frac{1}{n^{1/2}}\right) \nonumber \\
&=& p\left\{E\tilde m_2^2(X_2) - [E\tilde m_2(X_2)]^2\right\} + O_p\left(\frac{1}{n^{1/2}}\right), \nonumber
\end{eqnarray}
it follows that
\begin{eqnarray}
\sqrt{n}\rho_n^2 \tilde m_{2V}'A\tilde m_{2V} &=& a^2p\Var(\tilde m_2(X_2)) + O_p\left(\frac{1}{n^{1/2}}\right), \nonumber
\end{eqnarray}
which completes the proof.
\end{proof}
\vspace{-3mm}

\begin{lemma} \label{pf.General.Second} The second term in (\ref{general.terms}) converges in probability to $2 p ab Cov(\tilde m_2(X_2), \tilde m_{12}(\vX_1,X_2))$, i.e.
\vspace{-4mm}
 \bqa
&& \sqrt{n}2\rho_{2n} (\vxi_V - \vDelta_{m_1V} - \rho_{1n}\tilde \vm_{2V})'A\tilde \vm_{12V} \pkonv 2 p ab Cov(\tilde m_2(X_2), \tilde m_{12}(\vX_1,X_2)).
 \eqa
\end{lemma}
\begin{proof} By the definition of the matrix $A$, we can write
\bqa
&& \sqrt{n}\rho_{2n}(\vxi_V - \vDelta_{m_1V} - \rho_{1n}\tilde \vm_{2V})'A\tilde \vm_{12V} = \sqrt{n}\rho_{2n}\frac{np-1}{n(n-1)p(p-1)}\times \nonumber \\
&& \times \sum_{i=1}^{n}\left[\sum_{j=1}^{n}\tilde m_{12}(\vX_{1j},X_{2j})I(j \in W_i)\right]\left[\sum_{k=1}^{n}(\xi_k - \Delta_{m_1}(\vX_{1k}) - \rho_{1n}\tilde m_{2}(X_{2k}))I(k \in W_i)\right] \nonumber 
\eqa
\bqan 
&& - \sqrt{n}\rho_{2n}\frac{1}{n(n-1)p}\left[p\sum_{i=1}^{n}\tilde m_{12}(\vX_{1i},X_{2i})\right]\left[p\sum_{i=1}^{n}(\xi_i - \Delta_{m_1}(\vX_{1i}) - \rho_{1n}\tilde m_{2}(X_{2i}))\right] \nonumber \\
&& - \sqrt{n}\rho_{2n}\frac{p}{n(p-1)}\sum_{i=1}^{n}\tilde m_{12}(\vX_{1i},X_{2i})(\xi_i - \Delta_{m_1}(\vX_{1i}) - \rho_{1n}\tilde m_{2}(X_{2i})).  \label{rel.pf.General.Second.terms}
\end{eqnarray}
Noting that $n^{-3/4}\sum_{i=1}^{n}\xi_i \tilde m_{12}(\vX_{1i},X_{2i}) \pkonv 0$, and $\frac{1}{n^{3/4}}\sum_{i=1}^{n}\Delta_{m_1}(\vX_{1i})m_{12}(\vX_{1i},X_{2i}) \pkonv 0$, which follows by arguments similar to (\ref{rel.pf.2ndterm.rel.1st.term}), the third term in (\ref{rel.pf.General.Second.terms}) goes in probability to 
$[pab/(p-1)]E(\tilde m_2(X_2)\tilde m_{12}(\vX_1,X_2))$. Also, using (\ref{rel.pf.2ndterm.rel.1st.term}), and the facts  $E(\tilde m_{12}(\vX_1,X_2)) = 0$, and
$n^{-3/4}\sum_{i=1}^{n}\xi_i = o_p(1)$, the second term in (\ref{rel.pf.General.Second.terms}) goes to $pabE(\tilde m_2(X_2))E(\tilde m_{12}(\vX_1,X_2))$ in probability. Next, the component of the first term in (\ref{rel.pf.General.Second.terms}) that corresponds to
\bqa
\sum_{j=1}^{n}\sum_{k=1}^{n}\sum_{i=1}^{n}\tilde m_{12}(\vX_{1j},X_{2j})(\xi_{k} - \Delta_{m_1}(\vX_{1k}) )I(j \in W_i)I(k \in W_i)
\eqa
goes to zero in probability by arguments similar to those used for the last term in (\ref{rel.pf.General.Second.terms}). Set $\bar{m}_{2}^{i}(X_{2i})=\frac{1}{p}\sum_{j=1}^{n}\tilde m_{2}(X_{2j})I(j \in W_i) $ and $ \bar{m}_{12}^{i}(.,X_{2i})= \frac{1}{p}\sum_{j=1}^{n}\tilde m_{12}(\vX_{1j},X_{2i})I(j \in W_i) $, so that
 \bqa
&& \frac{1}{p}\sum_{j=1}^{n}\tilde m_{12}(\vX_{1j},X_{2j})I(j \in W_i) =  \bar{m}_{12}^{i}(.,X_{2i})+o_p(1),\\
&& \frac{1}{p}\sum_{j=1}^{n}\tilde m_{2}(X_{2j})I(j \in W_i) = \bar{m}_{2}^{i}(X_{2i})+o_p(1).
  \eqa
The remaining component of the first term in (\ref{rel.pf.General.Second.terms})
can be written as
 \bqa
&& \frac{(np-1)ab}{n(n-1)p(p-1)}\sum_{j=1}^{n}\sum_{k=1}^{n}\sum_{i=1}^{n}\tilde m_{12}(\vX_{1j},X_{2j})
 \tilde m_2(X_{2k})I(j \in W_i)I(k \in W_i)\\
&=&\frac{(np-1)pab}{(n-1)(p-1)}\frac{1}{n}\sum_{i=1}^{n}\bar{m}_{12}^{i}(.,X_{2i}) \bar{m}_{2}^{i}(X_{2i})
+o_p(1)\pkonv \frac{p^2b^2}{p-1}E\left[\tilde m_{12}(\vX_1,X_2)\tilde m_2(X_2)\right],
\eqa
completing the proof.
\end{proof}

\begin{lemma} \label{pf.General.Third} The third term in (\ref{general.terms}) converges in probability to $pb^2Var(\tilde m_{12}(\vX_{1},X_{2}))$, i.e.
 \bqa
 \sqrt{n}\rho_{2n}^2 \tilde \vm_{12V}'A\tilde \vm_{12V} \pkonv pb^2Var(\tilde m_{12}(\vX_{1},X_{2})).
 \eqa
\end{lemma}
\begin{proof} Note that we can write $\sqrt{n}\rho_{2n}^2\tilde \vm_{12V}'A\tilde \vm_{12V}$ as
\begin{eqnarray} \label{eq:m12VAm12}
   && \frac{(np-1)b^2}{n(n-1)p(p-1)} \sum_{i=1}^{n}\left[\sum_{j=1}^{n}\tilde m_{12}(\vX_{1j},X_{2j})I(j \in W_i)\right]\left[\sum_{k=1}^{n}\tilde m_{12}(\vX_{1k},X_{2k}))I(k \in W_i)\right] \nonumber \\
&& -\frac{pb^2}{n(n-1)}\left[\sum_{i=1}^{n}\tilde m_{12}(\vX_{1i},X_{2i})\right]\left[\sum_{i=1}^{n}\tilde m_{12}(\vX_{1i},X_{2i})\right] \nonumber \\
&& - \frac{pb^2}{n(p-1)}\sum_{i=1}^{n}\tilde m_{12}(\vX_{1i},X_{2i})^2.
\end{eqnarray}
Clearly, the third term in (\ref{eq:m12VAm12}) goes to $[pb^2/(p-1)]E[\tilde m_{12}(\vX_1,X_2)^2]$ in probability, and the second term in (\ref{eq:m12VAm12}) goes to $pb^2[E(\tilde m_{12}(\vX_1,X_2))]^2$ in probability. Using the same notation as in lemma \ref{pf.General.Second}, the first term in (\ref{eq:m12VAm12}) is equal to
\begin{eqnarray} 
 \frac{(np-1)pb^2}{n(n-1)(p-1)} \sum_{i=1}^{n}\left[\bar{m}^{i}_{12}(.,X_{2i})\right]^2 + o_p(1) \pkonv \frac{p^2b^2}{p-1}E[(\tilde m_{12}(\vX_{1i},X_{2i}))^2], \nonumber
\end{eqnarray}
completing the proof.
\end{proof}

\vspace{-6mm}

\renewcommand{\baselinestretch}{1.00}


\baselineskip=14pt
\begin{thebibliography}{60}

\bibitem{Sahalia2001} Aït-Sahalia, Y., Bickel, P. J. and Stoker, T.M. (2001).  Goodness-of-fit tests for kernel regression with an application to option implied volatilities. \emph{Journal of Econometrics}, 105, 363-412.

\bibitem{Abramovich} Abramovich, F.,Benjamini, Y., Donoho, D.L. and Johnstone,I.M. (2006). Adapting to unknown sparsity by controlling the false discovery rate.\emph{The Annals of Statistics}, 34, 584-653.

\bibitem{Papadatos} Akritas, M. G. and Papadatos, N. (2004). Heterocedastic One-Way ANOVA and Lack-of-Fit Tests. \emph{Journal of the American Statistical Association}, 99, Theory and Methods.

\bibitem{BenjaminiGavrilov} Benjamini, Y.; Gavrilov, Y. (2009). A Simple Forward Selection Procedure Based on False Discovery Rate Control. \emph{The Annals of Applied Statistics}, 3, 179-198.

\bibitem{BenjaminiHochberg} Benjamini, Y.; Hochberg, Y. (1995). Controlling the false discovery rate: a practical and powerful approach to multiple testing. \emph{Journal of the Royal Statistical Society B}, 57 (1): 289–300.

\bibitem{BenjaminiKY} Benjamini, Y., Krieger, A.M., Yekutieli, D. (2006). Adaptive Linear Step-up False Discovery Rate controlling procedures. \emph{Biometrika}, 93 (3): 491-507.

\bibitem{Benjamini2} Benjamini, Y. and Yekutieli, D. (2001). The control of the false discovery rate in multiple testing under dependency. \emph{The Annals of Statististics}, 29, 1165-1188.

\bibitem{Birge} Birge, L. and Massart, P. (2001) A generalized Cp criterion for Gaussian model. \emph{Technical report}, Lab. De Probabilities, Univ. Paris VI.
(http://www.proba.jussieu.fr/mathdoc/preprints/index.html\#2001)



\bibitem{Bunea} Bunea, F., Wegkamp, M. and Auguste, A. (2006). Consistent variable selection in high dimensional regression via multiple testing. \emph{Journal of Statistical Planning and Inference,} 136, 4349-4364.


\bibitem{Candes} Candes, E., and Tao, T. (2007). The Dantzig selector: Statistical estimation when p is much larger than n. \emph{The Annals of Statistics}, 35, 2313-2351.

\bibitem{Delgado2001} Delgado, M. A. and Manteiga, W. G. (2001). Significance Testing in Nonparametric Regression Based on the Bootstrap \emph{The Annals of Statistics}, 29, 1469–1507.


\bibitem{DonohoJohnstone1994} Donoho, D. L. and Johnstone, I. M. (1994). Ideal spatial adaptation by wavelet shrinkage. \emph{Biometrika}, 81, 425–55.


\bibitem{Dvoretzky} Dvoretzky, A.; Kiefer, J. and Wolfowitz, J. (1956). Asymptotic minimax character of the sample distribution function and of the classical multinomial estimator. \emph{The Annals of Mathematical Statistics}, 27 (3), 642–669.

\bibitem{Efron} Efron, B., Hastie, T., Johnstone, I. and Tibshirani, R. (2004) Least angle regression. \emph{The Annals of Statistics}, 32, 407-499.



\bibitem{Fan2} Fan, J. and Jiang, J.(2005). Nonparametric Inferences for Additive Models. \emph{Journal of the American Statistical Association}, 100, 890-907.

\bibitem{Fan3} Fan, Y. and Li, Q. (1996). Consistent model specification tests: Omitted variables and semiparametric functional forms. \emph{Econometrica}, 64 (4), 865-890.

\bibitem{Fan4} Fan, J. and Li, R. (2001). Variable selection via nonconcave penalized likelihood and its oracle properties. \emph{Journal of the American Statistical Association}, 96, 1348-1360.

\bibitem{FZZ01} Fan, J., Zhang, C. M., and Zhang, J. (2001), Generalized Likelihood Ratio
Statistics and Wilks Phenomenon. {\it The Annals of Statistics}, 29, 153193.

\bibitem{FosterStein} Foster, D. P. and Stine, R. A. (2004). Variable selection in data mining: building a predictive model for bankruptcy. \emph{Journal of the American Statistical Association}, 99, 303-313.

\bibitem{Hansen} Hansen, B. E., (2008). Uniform convergence rates for kernel estimation with dependent data. \emph{Econometric Theory}, 24, 726-748.


\bibitem{Hart} Hart, J. D. (1997). {\it Nonparametric Smoothing and Lack of fit tests}, Springer.


\bibitem{HorowitzMammen} Horowitz, J.L. and Mammen, E. (2004). Nonparametric estimation of an additive model with a link function. {\it The Annals of Statistics}, 32, 2412-2443.

\bibitem{HHW2010} Huang, J., Horowitz, J. L., and Wei, F. (2010). Variable Selection in Nonparametric Additive Models. Available at
http://faculty.wcas.northwestern.edu/~jlh951/papers/HHW-npam.pdf


\bibitem{Lavergne} Lavergne, P. and Vuong, Q. (2000). Nonparametric Significance Testing. \emph{Econometric Theory}, 16, 576-601.

\bibitem{Li1991} Li, K. C., (1991). Sliced Inverse Regression for Dimension Reduction. \emph{Journal of the American Statistical Association}, 86, 316-327.

\bibitem{Li} Li, R. and Liang, H. (2008). Variable selection in Semiparametric Regression Modeling. \emph{The Annals of Statistics}, 36,
261-286.


\bibitem{LintonNiel95} Linton, O. and Nielsen, J.P. (1995). A kernel method of estimating structured nonparametric regression based on marginal integration. {\it Biometrika}, 82, 93-100. 

\bibitem{MLN99} Mammen, E., Linton, O. and Nielsen, J. (1999). The existence and asymptotic properties of a backfitting projection algorithm under weak conditions. {\it The Annals of Statistics}, 27, 1443-1490.



\bibitem{Newey} Newey, W. K. (1994). Kernel estimation of partial means. {\it Econom. Theory}, 10, 233-253.


\bibitem{RHL06} Racine, J., Hart, J.D. and Li, Q. (2006). Testing the significance of categorical predictor variables in
nonparametric regression models. {\it Econometric Reviews}, 25, 523-544.

\bibitem{Rice} Rice, J. (1984). Bandwidth choice for nonparametric regression. \emph{The Annals of Statistics}, 12, 1215-1230.

\bibitem{Ruschendorf} Ruschendorf, L. (1977). Consistency of Estimators for Multivariate Density Functions and For the Mode. \emph{Sankhya: The Indian Journal of Statistics}, 39, 243-250.

\bibitem{Stone} Stone, C. (1982). Optimal global rates of convergence for nonparametric regression. \emph{The Annals of Statistics}, 10, 1040-1053.

\bibitem{Stone} Stone, C. (1985). Additive regression and other nonparametric models. \emph{The Annals of Statistics}, 13, 689-705.

\bibitem{SBRZ2011} Storlie, C. B, Bondell, H. D, Reich, B. J, Zhang, H. H. (2011). Surface Estimation, Variable Selection, and the Nonparametric Oracle Property. \emph{Statistica Sinica}, 21(2), 679-705.


\bibitem{Tibishirani} Tibshirani, R. (1996). Regression shrinkage and selection via the lasso. \emph{Journal of the Royal Statististical Society B}, 58, 267-288.

\bibitem{Tibishirani2} Tibshirani, R. and Knight, K. (1999). The covariance inflation criterion for adaptive model selection. \emph{Journal of the Royal Statistical Society B}, 61, 529-546.



\bibitem{WX2008} Wang, H. and Xia, Y. (2008). Shrinkage estimation of the varying coefficient model. {Journal of the American Statistical Association}, 104,747-757.

\bibitem{Wang} Wang, L., Akritas, M. G. and Keilegom, I.V. (2008). An ANOVA-type Nonparametric Diagnostic Test for Heterocedastic Regression Models. \emph{Journal of Nonparametric Statistics}, 20, 365-382.


\bibitem{Zou} Zou, H. (2006). The Adaptive Lasso and its Oracle Properties. \emph{Journal of the American Statistical Association}, 101(476), 1418-1429.
\end{thebibliography}
\end{document}